\documentclass[11pt]{article}

\usepackage{booktabs}
\usepackage{amsmath, amssymb, amsthm, thmtools, amsfonts}
\usepackage{bbm}
\usepackage[utf8]{inputenc}
\usepackage{hyperref}
\usepackage[margin=1in]{geometry}

\usepackage{ifthen}
\usepackage{tikz}
\usetikzlibrary{positioning,decorations.pathreplacing}
\usepackage{appendix}
\usepackage{graphicx}
\usepackage{adjustbox}
\usepackage{array}
\usepackage{makecell}

\usepackage{color}
\usepackage{pagecolor}
\usepackage[ruled,linesnumbered]{algorithm2e}
\usepackage[noend]{algpseudocode}
\usepackage{epstopdf}
\usepackage[textsize=tiny]{todonotes}
 \usepackage{relsize}

\newcommand\eps{\varepsilon}

\usepackage{framed}
\usepackage[framemethod=tikz]{mdframed}
\usepackage[bottom]{footmisc}
\usepackage[shortlabels]{enumitem}
\setitemize{noitemsep,topsep=3pt,parsep=3pt,partopsep=3pt}
\usepackage[font=small]{caption}
\usepackage{xspace}

\newtheorem{theorem}{Theorem}[section]
\newtheorem{lemma}[theorem]{Lemma}
\newtheorem{fact}[theorem]{Fact}
\newtheorem{meta-theorem}[theorem]{Meta-Theorem}

\newtheorem{remark}[theorem]{Remark}
\newtheorem{corollary}[theorem]{Corollary}

\newtheorem{definition}[theorem]{Definition}

\newtheorem*{thm:msf_algo}{Theorem~\ref{thm:msf_algo}}
\newtheorem*{thm:sf_algo}{Theorem~\ref{thm:sf_algo}}

\definecolor{darkgreen}{rgb}{0,0.5,0}
\definecolor{darkred}{rgb}{0.9,0,0}
\definecolor{grayish}{rgb}{0.8,0.8,0.8}
\usepackage{hyperref}
\hypersetup{
    unicode=false,          
    colorlinks=true,        
    linkcolor=darkred,          
    citecolor=darkgreen,        
    filecolor=magenta,      
    urlcolor=cyan           
}
\usepackage[capitalize, nameinlink]{cleveref}
\crefname{theorem}{Theorem}{Theorems}
\Crefname{lemma}{Lemma}{Lemmas}
\Crefname{claim}{Claim}{Claims}
\Crefname{fact}{Fact}{Facts}
\algnewcommand\algorithmicswitch{\textbf{switch}}
\algnewcommand\algorithmiccase{\textbf{case}}

\algdef{SE}[SWITCH]{Switch}{EndSwitch}[1]{\algorithmicswitch\ #1\ \algorithmicdo}{\algorithmicend\ \algorithmicswitch}%
\algdef{SE}[CASE]{Case}{EndCase}[1]{\algorithmiccase\ #1}{\algorithmicend\ \algorithmiccase}%
\algtext*{EndSwitch}%
\algtext*{EndCase}%

\newcommand{\bigO}{\mathcal{O}}

\renewcommand{\paragraph}[1]{\vspace{0.15cm}\noindent {\bf #1}:}

\mathchardef\mhyphen="2D

\newcommand{\MPC}{$\mathsf{MPC}$\xspace}

\newcommand{\cclique}{Congested Clique\xspace}
\newcommand{\MST}{$\mathsf{MST}$\xspace}

\newcommand{\congest}{$\mathsf{CONGEST}$\xspace}

\newcommand{\degree}{\mathsf{degree}\xspace}
\newcommand{\ID}{\mathsf{ID\xspace}}

\newcommand{\set}[1]{\left\{#1\right\}}
\newcommand{\paren}[1]{\mathopen{}\left(#1\right)\mathclose{}}

\newcommand{\ceil}[1]{\mathopen{}\left\lceil#1\right\rceil\mathclose{}}

\newcommand{\card}[1]{\mathopen{}\left|#1\right|\mathclose{}}

\renewcommand{\paragraph}[1]{\vspace{0.15cm}\noindent {\bf #1}:}

\newcommand{\FullOrShort}{full}

\ifthenelse{\equal{\FullOrShort}{full}}{
    
  \newcommand{\fullOnly}[1]{#1}
  \newcommand{\shortOnly}[1]{}
  
  }{

    \newcommand{\fullOnly}[1]{}
    \newcommand{\IncludePictures}[1]{}
   
  }

\begin{document}

\date{}

\title{A Deterministic Algorithm for the MST Problem \\ in Constant Rounds of Congested Clique}

\author{\emph{Krzysztof Nowicki}\thanks{This research is supported by the Polish National Science Centre, under projects
number 2017/25/B/ST6/02010 and 2019/32/T/ST6/00566.}\\
\small University of Wrocław\\
\small \texttt{knowicki@cs.uni.wroc.pl}
}
\setcounter{page}{0}
\thispagestyle{empty}

\maketitle

\begin{abstract}
In this paper we show that the Minimum Spanning Tree problem (MST) can be solved \emph{deterministically} in $\mathcal{O}(1)$ rounds of the $\mathsf{Congested}$ $\mathsf{Clique}$ model. 

In the $\mathsf{Congested}$ $\mathsf{Clique}$ model there are $n$ players that perform computation in synchronous rounds. Each round consist of a phase of local computation and a phase of communication, in which each pair of players is allowed to exchange $\mathcal{O}(\log n)$ bit messages. The studies of this model began with the MST problem: in the paper by Lotker et al.[SPAA'03, SICOMP'05] that defines the $\mathsf{Congested}$ $\mathsf{Clique}$ model the authors give a deterministic $\mathcal{O}(\log \log n)$ round algorithm that improved over a trivial $\mathcal{O}(\log n)$ round adaptation of Borůvka's algorithm.

There was a sequence of gradual improvements to this result: an $\mathcal{O}(\log \log \log n)$ round algorithm by Hegeman et al. [PODC'15], an $\mathcal{O}(\log^* n)$ round algorithm by Ghaffari and Parter, [PODC'16] and an $\mathcal{O}(1)$ round algorithm by Jurdziński and Nowicki, [SODA'18], but all those algorithms were randomized, which left the question about the existence of any deterministic $o(\log \log n)$ round algorithms for the Minimum Spanning Tree problem open.

Our result resolves this question and establishes that $\mathcal{O}(1)$ rounds is enough to solve the MST problem in the $\mathsf{Congested}$ $\mathsf{Clique}$ model, even if we are not allowed to use any randomness. Furthermore, the amount of communication needed by the algorithm makes it applicable to some variants of the $\mathsf{MPC}$ model.

\end{abstract}

\bigskip
\setcounter{page}{0}
\thispagestyle{empty}
\newpage

\section{Introduction and related work}\label{s:intro}
In this paper, we present a simple deterministic algorithm for the Minimum Weight Spanning Tree problem (\MST) that needs only a constant number of rounds of \cclique. 

In \emph{the Minimum Weight Spanning Tree problem}, for a connected weighted input graph we have to compute the lightest acyclic set of edges that connects all vertices of the input graph. Our result also applies to \emph{the Minimum Weight Spanning Forest problem}, in which the input graph may be not connected and as a result we need to identify a minimum weight spanning tree of each connected component of the input graph. 

This is one of the central problems in graph algorithmics and solving it is used as a subroutine in many more complicated algorithms. The studies on this problem began over 90 years ago and the first algorithm for was proposed by Otakar Borůvka in 1926 \cite{Boruvka}.

The \cclique model was introduced by Lotker et al. \cite{lotker2003mst,DBLP:journals/siamcomp/LotkerPPP05} in a paper that also studies the \MST problem -- more precisely the authors study the \MST problem in a distributed model of computing (\congest), with an assumption that the communication network is a clique -- the name \cclique was coined only in some later papers. This makes the \MST problem not only a possibly useful subproblem to be solved, but also a core problem that was studied basically since the conception of the \cclique model \cite{DBLP:journals/siamcomp/LotkerPPP05,Hegeman:2015:TOB:2767386.2767434,Ghaffari:2016:MLR:2933057.2933103,DBLP:journals/corr/Korhonen16,DBLP:conf/soda/Jurdzinski018}.

\subsection{\cclique model}
The \cclique is a model of distributed (or parallel) computation, in which we have $n$ players (processors) performing computation in synchronous rounds. Each player corresponds to a single vertex of the input graph, and initially knows all edges that are incident to this vertex.

A single round consists of a phase of local computation, in which all players simultaneously perform computation, and a phase of communication, in which each pair of players can simultaneously exchange a pair of messages of size $\bigO(\log n)$ bits. In other words, it is a synchronous message passing model, in which the communication graph is a clique, the communicating players send and receive messages simultaneously, and the number of bits in a single message is $\bigO(\log n)$. 

As a result, we usually require that some player knows the answer computed by the algorithm. The exception are the problems for which the answer is too large to fit into memory of a single processor; then we allow that each player needs to know only a part of the result, but there are no general rules for this kind of problems and the specific requirements towards the output are usually tailored to the problem.

The \cclique model initially was considered as a special case of the \congest model in which the players may communicate only if in the input graph there is an edge between vertices corresponding to those players. The studies of \cclique began with adding an assumption that the input graph and, more importantly, the communication network in the \congest is a clique. Currently, the widely accepted definition says that in \cclique the communication network is a $n$ node clique and the input is an arbitrary $n$ vertex graph. 

Besides being just a special case of the \congest model, \cclique might be used to model the \emph{overlay networks} and has ties to some models of parallel computing.

\paragraph{Overlay networks}
Lotker et al. \cite{lotker2003mst,DBLP:journals/siamcomp/LotkerPPP05} propose that \cclique may be a good theoretical model to study the overlay networks: an abstraction that separates the problems emerging form the topology of the communication network from the problems emerging from the structure of the problem we try to solve. In other words, it allows us to study a model in which each pair of nodes can communicate, and we do not consider any details of how this communication is executed by the underlying network.

\paragraph{Parallel Computing}
\cclique can be also considered as a model of parallel computing, in particular, it is closely related to the \textbf{M}assively \textbf{P}arallel \textbf{C}omputation (\MPC) model \cite{karloff2010MapReduce,DBLP:conf/sirocco/HegemanP14,DBLP:journals/corr/abs-1802-10297}. In \MPC the computation is performed by a set of machines in synchronous rounds; each round consisting of a phase of local computation and a phase of communication. In the communication phase each pair of machines can exchange some number of messages, as long as each machine sends and receives a number of messages bounded by some parameter $S$, and total communication is bounded by $\bigO(N)$, where $N$ is the size of the input.

In the \cclique model, the restrictions on the number of messages that can be exchanged between the processors are stronger, as each pair of processors can exchange only one message. The problem of exchanging larger amounts of messages is called \emph{the routing problem}, and some of its variants can be solved in $\bigO(1)$ rounds, deterministically, by the routing protocol proposed by Lenzen \cite{Lenzen:2013:ODR:2484239.2501983}. More precisely, we can route all messages to their destinations whenever this problem looks like the communication in the \MPC model, i.e. when each processor is a source and destination of $\bigO(n)$ messages. This immediately gives that one can simulate some variants of the \MPC model (with $S \in \bigO(n)$) in \cclique. 

In the remaining part of the paper, we use the Lenzen's routing implicitly in several places, i.e. we show that in order to perform some task, each machine needs to send and receive a batch of $\bigO(n)$ messages, therefore exchanging all messages and performing this task can be done in $\bigO(1)$ rounds.

The connection between the \cclique and \MPC is bidirectional, as any algorithm that has \emph{small communication} in the \cclique model can be applied to the \MPC model \cite{DBLP:conf/sirocco/HegemanP14,DBLP:journals/corr/abs-1802-10297}, for properly defined notion of \emph{small communication}. The routing protocol of Lenzen on its own requires $\Theta(n^2)$ messages, no matter what is the number of messages to be send. Therefore, each algorithm that uses it needs $\Theta(n^2)$ messages in \cclique to be exchanged. On the other hand, in the \MPC model we get communication that is handled by the Lenzen's routing protocol for free. Thus, by \emph{small communication} we mean that in a single round of the \cclique algorithm the total number of messages that are send directly is $\bigO(n)$ per vertex and $O(N)$ in total, and the total number of messages that are send and received via routing protocol is $\bigO(n)$ per vertex and $O(N)$ in total. In other words, for the purpose of this paper we decided to measure the communication complexity of the \cclique algorithms without taking into account the cost of the routing protocol by Lenzen.

\subsection{Minimum Spanning Tree problem in \cclique} \label{sec:related} 
The MST problem was studied already in the seminal paper \cite{DBLP:journals/siamcomp/LotkerPPP05} that introduced the \cclique model. In \cite{DBLP:journals/siamcomp/LotkerPPP05} Lotker et al. study the MST problem in a distributed \congest network of diameter $1$ and propose an algorithm that needed only $\bigO(\log \log n)$ rounds to determine the MST of the input graph. 

The first improvement to this results came only after around 10 years, when randomized graph sketching techniques \cite{ahn2012analyzing} were applied by Hegeman et al. \cite{Hegeman:2015:TOB:2767386.2767434} to obtain an $\bigO(\log \log \log n)$ round algorithm. This paper also established a $\bigO(1)$ round, randomized reduction that reduces a single instance of the MST problem to several instances of the Connected Component problems. The two next papers improved the complexity of MST algorithm by providing better Spanning Forest algorithms that can be run in parallel: in \cite{Ghaffari:2016:MLR:2933057.2933103} the authors propose a $\bigO(\log^* n)$ round algorithm, and finally \cite{DBLP:conf/soda/Jurdzinski018} shows that MST can be found in $\bigO(1)$ rounds.

All those $o(\log \log n)$ round algorithms are randomized and heavily rely on the sketching techniques. The only alternative deterministic algorithm for the MST problem, by Korhonen \cite{DBLP:journals/corr/Korhonen16}, also has $\bigO(\log \log n)$ bound on the round complexity. Therefore, the $\bigO(1)$ round deterministic algorithm we propose in this paper not only shows that the problem can be solved also in a deterministic way, but provides the first improvement in the complexity of deterministic MST algorithms since the beginning of studies on MST problem in \cclique model.
\section{Our results}\label{s:results}

The most significant result presented in this paper is a deterministic $\bigO(1)$ round algorithm for the Minimum Spanning Tree problem.

\begin{theorem}\label{thm:msf_algo}
The Minimum Spanning Tree problem can be solved deterministically in $\bigO(1)$ rounds of the \cclique model, using $\bigO(m)$ communication, $m$ is the number of edges of the input graph.
\end{theorem}
On the top level, the algorithm is based on the reduction from a single instance of MST to many instances of the Connected Components problem proposed by Hegeman et al. \cite{Hegeman:2015:TOB:2767386.2767434}. The only part of this reduction that is randomized algorithm that reduces a single instance of the MST problem, to two instances, each with $\bigO(n^{3/2})$ edges. Here, we use that one can replace the randomized part of this reduction with sparsification algorithm by Korhonen \cite{DBLP:journals/corr/Korhonen16}.
\begin{remark}In \cite{DBLP:journals/corr/Korhonen16} Korhonen shows that one can reduce an instance of the MST problem to another instance of MST problem, with $\bigO(n^{1+\eps})$ edges, in $\bigO(1)$ rounds of \cclique. Therefore, it was known that this sparsification technique can be used as a part of the Hegeman et al. reduction \cite{Hegeman:2015:TOB:2767386.2767434} for the \cclique model. For our purpose, we need to extend this result by:
\begin{itemize}
  \item an analysis of communication complexity that makes the sparsification technique applicable to \MPC,
  \item additional algorithmic tools that allow us to significantly sparsify the input graph in $\bigO(1)$ rounds, even if the input graph is already sparse.
\end{itemize}
\end{remark}

As in the case of \cite{Ghaffari:2016:MLR:2933057.2933103,DBLP:conf/soda/Jurdzinski018} the key contribution that leads to improvements for the MST problem is an algorithm that solves the Spanning Forest problem (and by extension the Connected Components problem) that can be run efficiently in parallel. 

\begin{theorem}\label{thm:sf_algo}
The Spanning Forest problem can be solved deterministically in $\bigO(1)$ rounds of the \cclique model, using $\bigO(m)$ communication, where $m$ is the number of edges of the input graph.
\end{theorem}
To obtain this result, we propose a rather straightforward extension of the sparsification technique by Korhonen \cite{DBLP:journals/corr/Korhonen16}, which we combine with an approach based on some newly discovered properties of a deterministic part of the algorithm proposed by Jurdziński and Nowicki \cite{DBLP:conf/soda/Jurdzinski018,jurdznski_et_al:LIPIcs:2017:7990}. We give an algorithm that is a proof of \cref{thm:sf_algo} in \cref{sec:sf_algo}. 

In our algorithm for the Minimum Spanning Tree problem, apart from the communication following from the Lenzen's routing protocol, every processor needs to send and receive only $\bigO(n)$ messages. Furthermore, the total number of messages that are exchanged in a single round is $\bigO(m)$. Therefore, our MST algorithms can be implemented in the \MPC model with $\bigO(n)$ memory per machine and $\bigO(m)$ global memory.

\begin{corollary}
The Minimum Spanning Tree problem can be solved deterministically in $\bigO(1)$ rounds of the \MPC model that uses $\bigO(m)$ global memory and $\bigO(n)$ memory per machine, where $n$ is the number of vertices and $m$ is the number of edges of the input graph. 
\end{corollary}

\subsection{Structure of the remaining part of this paper}
In \cref{sec:sparsify} we present a variant of Korhonen's sparsification technique that is adjusted to our needs. Then, in \cref{sec:sf_algo} we present the algorithm for the Spanning Forest problem. Finally, in \cref{sec:msf_algo} we briefly explain how the algorithm for the Spanning Forest problem can be run in parallel to fit into the reduction by Hegeman et al. \cite{Hegeman:2015:TOB:2767386.2767434}.

\section{Deterministic Sparsification in the \cclique}\label{sec:sparsify}
In this section we present a variant of the sparsification technique by Korhonen \cite{DBLP:journals/corr/Korhonen16} that can be applied to sparser graphs. Furthermore, we show that the communication complexity of this algorithm is $\bigO(m)$.

\begin{lemma}\label{lem:mst_sparsification}
  There is a deterministic, $\bigO(1)$ round \cclique algorithm, using $\bigO(m)$ messages that reduces an instance of the \MST problem on a graph with $n$ vertices and $m$ edges to an instance of the \MST problem that has $\bigO(n)$ vertices and $\bigO(\sqrt{mn})$ edges. 
\end{lemma}
The remaining part of this section is a proof of \cref{lem:mst_sparsification}. In \cref{subsec:korhonen} we recall the sparsification algorithm by Korhonen.  Then, in \cref{subsec:initial}, we show a preprocessing that allows us to use Korhonen's algorithm to get the result claimed in \cref{lem:mst_sparsification}.

\begin{remark}
Our variant of the sparsification algorithm, if executed as in the paper by Korhonen, needs only $\bigO(\log \log \Delta_A)$ rounds to complete computation, where $\Delta_A$ is an average degree of the input graph.
\end{remark}

\begin{remark}
The sparsification algorithm never uses that the input graph is connected, therefore it can be applied to the Minimum Spanning Forest problem. Furthermore, one can ignore the weights of the graph, which makes the sparsification algorithm applicable also to the Spanning Forest problem.
\end{remark}

\subsection{Deterministic Sparsification via $\Delta$--partitions} \label{subsec:korhonen}
A $\Delta$--partition \footnote{This notion of $\Delta$--partition corresponds to the notion of $\varepsilon$--partition from the paper by Korhonen, with $\Delta = n^{\varepsilon}$.} of the graph $G=(V,E)$is a partition of $V$ into disjoint sets $V_1, V_2, \dots, V_{\Delta}$, such that 
\begin{itemize}
  \item for each $i$, $\card{V_i} \in \bigO(n / \Delta)$,
  \item for each $i,j$\footnote{note that this also includes $i=j$}, $\card{\set{\set{u,v}| u \in V_i,\ v \in V_j, \set{u,v} \in E}} \in \bigO(n / \Delta)$.
\end{itemize}
 One of the useful properties of $\Delta$--partition is that a graph with $\Delta$-partition has only $\bigO(n \Delta)$ edges. This bound holds, because we have only $\Delta^2$ pairs of sets $V_i, V_j$, and for each such pair there are at most $\bigO(n / \Delta)$ edges, which in total gives $\Delta^2 \cdot \bigO(n / \Delta) = \bigO(n \Delta)$ edges.

The main contribution of the paper by the Korhonen \cite{DBLP:journals/corr/Korhonen16} is an $\bigO(1)$ round deterministic \cclique algorithm that given a graph with $\Delta$--partition computes a graph with a $\sqrt{\Delta}$--partition, while preserving all edges of the minimum spanning tree of the input graph. We state this algorithm as \cref{alg:korhonen}.
\begin{algorithm}
  \SetAlgoLined

  let $V'_i = \bigcup_{j=(i-1)\sqrt{\Delta}+1}^{i\sqrt{\Delta}}V_j$\\
  partition the edges in such a way that for all $i \leq j$ the edges $E_{i,j} = \set{\set{u,v}|\ u \in V'_i,\ v \in V'_j, \set{u,v} \in E}$ are in the memory of a single processor\\
  for all $i\leq j$ compute a minimum spanning forest $F_{i,j}$ of a graph consisting of edges $E_{i,j}$\\
  return graph $ (V,\bigcup_{i,j} \text{edges of }F_{i,j})$ with partition $V'_1, \dots, V'_{\sqrt{\Delta}}$\\
  \caption{Sparsify($G = (V,E)$, $\Delta$--partition)\cite{DBLP:journals/corr/Korhonen16}} \label{alg:korhonen}
\end{algorithm}

Firstly, we recall some properties of the sparsification technique by Korhonen (\cref{lem:sp_alg}), then we discuss its implementation in \cclique.

\begin{lemma}\cite{DBLP:journals/corr/Korhonen16}\label{lem:sp_alg}
  \cref{alg:korhonen} returns a graph $G'$ with a $\sqrt{\Delta}$--partition, such that the minimum spanning tree of $G'$ is also the minimum spanning tree of $G$. 
\end{lemma}
\begin{proof}
  To claim that we preserve minimum spanning tree we use the cycle property\cite{Tarjan:1983:DSN:3485}. This property says that any edge that is the heaviest edge on some cycle in a graph $G$ cannot be in the minimum spanning tree of $G$. Here, we firstly observe that if an edge $\set{u,v}$ does not belong to the minimum spanning forest of $E_{i,j}$, then $u$ and $v$ have to be connected over the edges of $E_{i,j}$, and $\set{u,v}$ is heavier than all the edges on the path connecting $u$ and $v$. Therefore, there exists a cycle in $G$ such that the edge $\set{u,v}$ is the heaviest edge in this cycle. Hence, $\set{u,v}$ cannot belong to the minimum spanning tree of $G$.
  
  To justify the claim that the obtained partition is a $\sqrt{\Delta}$--partition we only need to give a bound on the number of edges between the sets $V'_i, V'_j$ that are preserved. For a pair of sets $V'_i , V'_j$ we only keep the edges from $F_{i,j}$. Since $\card{V'_i \cup V'_j} \in \bigO(n / \sqrt{\Delta})$, the size of the spanning forest $F_{i,j}$ on vertices from $V'_i \cup V'_j$ is also $\bigO(n / \sqrt{\Delta})$. This concludes the proof of \cref{lem:sp_alg}.
\end{proof}

\paragraph{\cclique implementation of \cref{alg:korhonen}}
The partition of vertices into sets depends only on the identifiers of vertices, therefore it can be carried out locally. To compute the spanning forests of $E_{i,j}$, Korhonen proposed that each pair $i,j$ gets a dedicated vertex of the clique (we call such vertex a coordinator), which gathers all edges of $E_{i,j}$ and computes its minimum spanning forest in the local memory. To show that this implementation can be carried out, it is enough to show that $\card{E_{i,j}} \in \bigO(n)$. 

The set of vertices incident to edges of $E_{i,j}$ consists of $\bigO(\sqrt{\Delta})$ sets $V_{\alpha}, V_{\alpha+1}, \dots, V_{\alpha+\Theta(\sqrt{\Delta})}$, that are part of a $\Delta$--partition of the graph. By definition of $\Delta$--partition there are at most $\bigO(n / \Delta)$ edges between vertices from sets $V_{\alpha_1}, V_{\alpha_2}$, for any $\alpha_1, \alpha_2 \in \{\alpha, \alpha+1, \dots, \alpha+\Theta(\sqrt{\Delta})\}$. Therefore, $\card{E_{i,j}} \in \bigO((\sqrt{\Delta})^2) \cdot \bigO(n/\Delta) = \bigO(n)$.

Our observation is that, in order to carry out the implementation of this step, we need only $\bigO(m/n)$ coordinator vertices, rather than $\Theta(\Delta^2)$. The reason is that each edge of the graph ends up being a member of $E_{i,j}$ for exactly one pair $i,j$. Therefore, the total size of the sets of edges we have to gather is $\bigO(m)$, and the maximal size is still $\bigO(n)$. Hence, $\bigO(m/n)$ coordinator vertices are enough to store the edges in all sets $E_{i,j}$. The assignment of pairs $i,j$ to processors can be done, for example, by a parallel prefix computation, i.e. for each pair $i,j$ we compute $\sum_{(i',j') \leq (i,j)}\card{E_{i',j'}}$ which is enough to compute the $\ID$ of processor that needs to handle $E_{i,j}$.

\subsection{Obtaining a graph with an $\bigO(\frac{m}{n})$--partition}\label{subsec:initial}
In this subsection, we provide a simple preprocessing that transforms an $n$ vertex, $m$ edge graph $G$ to an $\bigO(n)$ vertex, $\bigO(m)$ edge graph $G'$ with an $\bigO(m/n)$--partition, such that knowing the edges of the MST of $G'$ allows us to identify the edges of the MST of $G$. Applying \cref{alg:korhonen} on $G'$ gives us a graph with an $\bigO(\sqrt{(m/n)})$--partition, hence having only $\bigO(\sqrt{mn})$ edges, which concludes the proof of \cref{lem:mst_sparsification}.

\paragraph{Initial reduction} Firstly, we transform a graph $G$ with $n$ vertices, $m$ edges, average degree $\Delta_A = 2m / n$ into a graph $G'$ with $\bigO(n)$ vertices, and maximal degree $\Delta_A+2$, in such a way that computing the edges of the MST of $G'$ allows to identify the edges of the MST of $G$. To obtain $G'$, we split each vertex with degree $\delta > \Delta_A$ into $\ceil{\delta/\Delta_A}$ vertices of degree at most $\Delta_A+2$, connected by a path (consisting of newly introduced \emph{path edges}). To each new vertex we assign at most $\Delta_A$ edges corresponding to the edges in the original graph and at most $2$ path edges. To the introduced path edges we assign a weight that is smaller than all weights in the input graph.

\paragraph{Initial reduction -- implementation} Here, we discuss an implementation of the initial reduction that is suitable for \cclique and \MPC. To obtain a partition of vertices into vertices of degree at most $\Delta_A+2$, it is enough to gather all degrees of vertices in the memory of a single processor. This processor then decides for each vertex what is the number of vertices it has to be splitted into, and assigns the IDs to the newly created vertices. We assign the new IDs in a way that each vertex gets splitted into several vertices that get new IDs that form a sequence of consecutive numbers. This allows to communicate the number of vertices and their identifiers as two messages: one that is the number of vertices to be created, and the other that is the smallest ID of a created vertex. Therefore, the total number of messages needed to be send by the processor that computes the splitting is $\bigO(n)$, and sending those messages to appropriate processors be done in $\bigO(1)$ rounds.

\paragraph{Initial reduction -- preserving MST} Here, we explain that we can compute the edges of the MST of $G$ out of the edges of the MST of $G'$. Let us consider an execution of Kruskal's algorithm on $G'$. The Kruskal's algorithm considers the edges from the lightest to the heaviest, hence it considers all the path edges created by the initial reduction before the edges that correspond to the edges of $G$. After processing all the path edges it computes a set of connected components that correspond to the vertices of $G$. The remaining edges correspond to the edges of $G$, and all edges included in the MST of $G'$ from this point correspond to the edges of the MST of $G$. In other words, removing all edges of the MST that are the path edges introduced by the initial reduction leaves only the edges that correspond to the edges of MST of $G$.

\paragraph{Initial reduction -- the number of vertices} Here, we show the bound on the number of vertices of $G'$. We can think that the protocol assigning the edges to the new vertices assigns them greedily, i.e. all but last vertex is incident to $\Delta_A$ edges corresponding to the edges of the input graph. Therefore, we have at most $n$ new vertices with degree $< \Delta_A$. Furthermore, having more than $2n$ new vertices of degree $\Delta_A$ would imply that in the original graph $G$ the sum of degrees had to be larger than $2n\Delta_A = 2n\frac{m}{n}$. This is impossible as the sum of degrees in any graph is $2m$. Therefore, in the obtained graph we have at most $3n$ vertices with maximal degree no larger than $\Delta_A+2$. 

\paragraph{Computing $\bigO(\frac{m}{n})$--partition}
Here, we provide an algorithm that computes a $\bigO(\frac{m}{n})$--partition for the graph $G'$ obtained by the initial reduction. To that end, we use a slightly modified variant of \cref{alg:korhonen}. As an input we take a graph with degree bounded by $\Delta_A+2$, and we change the first line to define sets $V'_i$ as an arbitrary partition of $V$, such that each $|V'_i| \in \bigO(n/\Delta_A)$. The remaining part of the algorithm remains unchanged. The claim is that executing this variant of the algorithm on $G'$ gives us a graph with $\paren{\bigO(\Delta_A) =  \bigO(\frac{m}{n})}$-- partition.

The guarantees for the resulting graph follow from exactly the same analysis as that we have for \cref{alg:korhonen}. To show that the algorithm can be implemented, it is enough to show that  $\card{E_{i,j}} \in \bigO(n)$. Since we required that for each $i$, $\card{V_i} \in \bigO(n/\Delta)$, and the maximal degree is $\Delta_A+2$, the total number of edges incident to vertices in $V_i \cup V_j$ is bounded by $\bigO(n/\Delta_A) \cdot (\Delta_A+2) \in \bigO(n)$. Since $E_{i,j}$ consists only of the edges that are incident to $V_i \cup V_j$, $\card{E_{i,j}} \in \bigO(n)$. The argument that explains why all edges of the MST are preserved remains unchanged.

The only additional remark, regarding the algorithm that computes an $\bigO(\frac{m}{n})$--partition for $G'$, is that a single processor may simulate several vertices of $G'$. This could potentially lead to the case in which a single processor has to send or receive $\omega(n)$ messages. However, the number of messages that have to be sent to the coordinators and received from the coordinators is bounded by the degree of the vertices that are simulated by a single processor. Since the sum of degrees of all vertices simulated by a single vertex of degree $\delta$ is at most $\delta + 2 \cdot \delta / \Delta_A \leq 3\delta$, the overall number of messages to be send by a single vertex increases only by a constant factor, and communication still can be executed in $\bigO(1)$ rounds.

\section{Deterministic algorithm for the Spanning Forest problem}\label{sec:sf_algo}
In this section we propose an $\bigO(1)$ round deterministic algorithm that solves the Spanning Forest problem in \cclique and \MPC models.

\begin{thm:sf_algo}
The Spanning Forest problem can be solved deterministically in $\bigO(1)$ rounds of the \cclique model, using $\bigO(m)$ communication, where $m$ is the number of edges of the input graph.
\end{thm:sf_algo}

The remaining part of this section contains a proof of \cref{thm:sf_algo}. In the paper \cite{DBLP:conf/soda/Jurdzinski018}, we apply \cref{lem:deg} to reduce a single instance of a Spanning Forest problem to two instances:
\begin{itemize}
\item an instance that consists only from vertices that, in the input graph, have degree smaller than $s$ (although it does not necessarily contains all such vertices),
\item an instance for which we know a partition into at most $n/s$ connected components. 
\end{itemize}
For graphs that are almost regular, i.e., in which all vertices have degree $\Theta(\delta)$, for some parameter $\delta$, an algorithm based on \cref{lem:deg} together with sparsification algorithm from \cref{lem:mst_sparsification} can solve the Spanning Forest problem. 

Using the algorithm based on \cref{lem:deg} we can compute a partition into $\Omega(n/\delta)$ components. This is because setting $s \in \Theta(\delta)$ to be smaller than minimum degree leaves the first instance empty, and for the second instance it gives a partition into $\bigO(\frac{n}{\delta})$ connected components. 

\begin{definition}
Let $\mathcal{C} = {C_1, C_2, \dots}$ be a partition of vertices of a graph $G=(V,E)$ into connected components. The component graph $G_\mathcal{C}$ is a graph in which the set of vertices corresponds to the set of components from $\mathcal{C}$, and set of edges consists of edges of $G$ that are between the components from $\mathcal{C}$, that is for each edge $\set{u,v}$ such that $u \in C_i$ and $v \in C_j$ such that $i \neq j$ there is an edge between the vertices of $G_\mathcal{C}$ corresponding to $C_i$ and $C_j$. 
\end{definition}

Let $\mathcal{C_{\delta}}$ is a partition into components obtained by the algorithm based on \cref{lem:deg}, $G_{C_{\delta}}$ is a graph with $\bigO(n / \delta)$ vertices and $\bigO(n\delta)$ edges. Therefore, an application of \cref{lem:mst_sparsification} on $G_{C_{\delta}}$ gives us a graph with $\bigO\paren{\sqrt{n\delta\frac{n}{\delta}}} = \bigO(n)$ edges. Such graph can be gathered in the local memory of a single processor, and this processor can compute a spanning forest locally.

In this section, we show how to extend this approach to handle graphs that are not necessarily almost regular. In \cref{subsec:prep} we give a statement and a proof of \cref{lem:deg}. Then, in \cref{subsec:sf_algo}, we present a few observations about the graph obtained by an application of \cref{lem:deg}. Those observations, when combined with sparsification algorithm from \cref{lem:mst_sparsification}, prove \cref{thm:sf_algo}.

\subsection{A technique reducing the number of components} \label{subsec:prep}
In this subsection we recall a simple lemma from \cite{DBLP:conf/soda/Jurdzinski018,jurdznski_et_al:LIPIcs:2017:7990} that allows us to compute a partition into connected components with the following property: a vertex of degree $\delta$ is a member of a connected components of size at least $\delta+1$.
 
The algorithm that computes such partition has two stages. In each stage, for each vertex we choose a single edge. In the first stage, for each vertex $v$ we select an edge connecting $v$ to a neighbour with the highest degree. Then, in the second stage, if there are some edges incident to $v$ that were not used in the first stage, we select for each $v$ one of those edges. Then we compute connected components of a graph consisting of selected edges. We state a more precise formulation of this algorithm as \cref{alg:reduce}. 

\begin{algorithm}
  \SetAlgoLined
  \KwIn{A graph $G$}
  \KwOut{A set of connected components $\mathcal{C}$}
  each vertex $v$ marks an edge connecting it to a neighbour with the highest degree (break ties towards higher $\ID$)\\
  each vertex $v$ notifies all neighbours, whether the edge between them was marked\\
  each vertex $v$ marks an edge connecting it to a vertex $u$ that did not mark the edge $\set{u,v}$ (if such $u$ exists)\\
  each vertex $v$ sends the marked edges to the coordinator vertex \\
  the coordinator vertex computes the connected components using gathered edges\\
  \caption{REDUCE COMPONENTS}\label{alg:reduce}
\end{algorithm}

\begin{lemma}\cite{DBLP:conf/soda/Jurdzinski018,jurdznski_et_al:LIPIcs:2017:7990}\label{lem:deg}
  After execution of \cref{alg:reduce}, a vertex that has degree $\delta$ becomes a member of a component of size at least $\delta+1$.
\end{lemma}
\begin{proof}
Consider a vertex $v$ and let $u$ be the vertex with the lexicographically largest $(\degree(u),$ $\ID(u))$ in the connected component of node $v$. We claim that all neighbors of $u$ in the original graph are in the same connected component. Otherwise, $u$ has neighbors that did not choose $u$ in the first step; let $w$ be the neighbor among these that $u$ chose in the second step. Node $w$ chose to connect to some other vertex $u'$ such that $(\degree(u'), \ID(u')) > (\degree(u), \ID(u))$. But now $u$ is connected to $u'$ and the existence of such a node $u'$ in this component is in contradiction with the choice of $u$. Thus, all neighbors of $u$ are in the same component, which means that this component has at least $\degree(u)+1 \geq \degree(v)+1$ vertices.
\end{proof}

\subsection{Beyond the almost regular graphs} \label{subsec:sf_algo}
In this subsection we propose an algorithm for the Spanning Forest problem. It consists of three main parts. \begin{itemize}
\item The first part of the spanning forest algorithm is to run \cref{alg:reduce} on the input graph $G$ to obtain a set of connected components $\mathcal{C}$. 
\item In the second part, the algorithm computes a partition of the component graph $G_\mathcal{C}$ into edge disjoint graphs $G_1, G_2, \dots$ that have some desired properties, which we define in the later part of this subsection. Then, the spanning forest algorithm executes the sparsification algorithm from \cref{lem:mst_sparsification} on graphs $G_i$, for all $i$ in parallel. 
\item Finally, in the third part, the algorithm gathers the edges that span the components from $G_\mathcal{C}$ together with all remaining inter component edges in the memory of a single processor, and this processor then computes the spanning forest of $G$. 
\end{itemize}
We discuss the first part in \cref{subsec:prep} and the sparsification algorithm in \cref{subsec:korhonen}. Here we focus on putting those building blocks together. We provide a pseudocode of the Spanning Forest algorithm in \cref{alg:sf_algo}.

\begin{algorithm}[H]
  \SetAlgoLined
  \KwIn{A graph $G$}
  \KwOut{A spanning forest of $G$}
  \SetAlgoLined
  $\mathcal{C} \gets$ run Reduce\_components(G)\\
  partition $G_{\mathcal{C}}$ into edge disjoint $G_1, G_2, \dots$ with properties stated as \cref{fact:edges} and \cref{fact:vertices}
  apply sparsification alg. from \cref{lem:mst_sparsification} on $G_i$, for all $i$ in parallel\\
  gather all remaining inter component edges together with the edges used for computing $\mathcal{C}$ in a memory of the coordinator\\
  the coordinator computes the spanning forest using gathered edges\\
  \caption{SPANNING FOREST}\label{alg:sf_algo}
\end{algorithm}

The algorithm we propose is based on the intuition that applying \cref{lem:deg} should cause a significant reduction of the number of connected components in some sufficiently dense subgraphs. We show that one can partition the edges of $G_{\mathcal{C}}$ in such a way that, after an execution of the sparsification algorithm from \cref{lem:mst_sparsification} on each part of the partition
\begin{itemize}
\item we obtain a graph with $\bigO(n)$ edges in total
\item all edges of some spanning forest of the input graph are preserved
\end{itemize}

Let $\mathcal{C} = C_1, C_2, \dots$ are the connected components obtained by \cref{alg:reduce} applied on a graph $G$. We define:
\begin{itemize}
\item $|C_i|$ to be a weight of vertex $v_i \in G_{\mathcal{C}}$ that corresponds to $C_i$,
\item $V_j$ as a set of vertices of $G_{\mathcal{C}}$ of weight at least $2^{j-1}$ and less than $2^j$,
\item $x_j$ be a sum of weights of vertices in $V_j$,
\item $y \in \bigO(\log n)$ be the maximal index of a non empty $V_j$.
\end{itemize}

Let us consider graphs $G_1, \dots, G_y$, where $G_i$ is defined as follows. The set of vertices of $G_i$ consists of vertices in $V_i$ and all vertices from $\bigcup_{j=i}^{y} V_j$ that are neighbours of vertices of $V_i$. The set of edges of $G_i$ consists of the edges of $G$ have at least one endpoint in $V_i$ and other in $\bigcup_{j=i}^{y} V_j$. Below we make two observations regarding graphs $G_i$, stated as \cref{fact:edges} and \cref{fact:vertices}.
\begin{fact} \label{fact:edges}
The number of edges of $G_i$ is smaller than $x_i2^i$.
\end{fact}
\begin{proof}
By definition, each edge of $G_i$ has an endpoint in set $V_i$, which consists of vertices that have weight less than $2^i$. In other words, each vertex in $V_i$ corresponds to a component computed by \cref{alg:reduce} that has size less than $2^i$. By \cref{lem:deg}, all vertices that are in such components have degree smaller than $2^i$. The number of the vertices of the original graph that form the components corresponding to the vertices in $V_i$ is $x_i$. Therefore, the total number of edges incident to those vertices is smaller than $x_i2^i$.
\end{proof}

\begin{fact} \label{fact:vertices}
The number of vertices of $G_i$ is at most $\frac{1}{2^i}\sum_{j=i}^{y} x_j/2^{j-1-i}$.
\end{fact}
\begin{proof}
By definition, the vertices of $V_i$ have weight at least $2^{i-1}$, which means that each vertex from $V_i$ corresponds to a component computed by \cref{alg:reduce} of size at least $2^{i-1}$. Since $x_i$ is exactly the number of the vertices of the original graph that form components corresponding to the vertices in $V_i$, the total number of vertices of $G_{\mathcal{C}}$ in $V_i$ is at most $x_i / 2^{i-1}$. By definition, the set of vertices of $G_i$ consists only of vertices from $\bigcup_{j=i}^{y} V_j$, hence it cannot be larger than $\sum_{j=i}^{y} x_j/2^{j-1} = \frac{1}{2^i}\sum_{j=i}^{y} x_j/2^{j-1-i}$.
\end{proof}

\paragraph{Reduction of the number of edges}
After execution of \cref{alg:reduce}, \cref{alg:sf_algo} obtains a set of $\bigO(\log n)$ graphs $G_1, G_2, \dots, G_y$ with properties stated as \cref{fact:edges} and \cref{fact:vertices}. The next step of \cref{alg:sf_algo} is to execute the algorithm from \cref{lem:mst_sparsification} to all graphs $G_1, G_2, \dots, G_y$, in parallel. Let $G_1^R, G_2^R, \dots, G_y^R$ be a set of obtained graphs. 

\begin{lemma}\label{lem:edge_reduce}
The total number of edges in $G_1^R, G_2^R, \dots, G_y^R$ is $\bigO(n)$.
\end{lemma}
\begin{proof}
Using \cref{fact:edges,fact:vertices} we have that an execution of the algorithm from \cref{lem:mst_sparsification} on $G_i$, gives a reduced graph $G_i^R$ with the number of edges that can be bounded by $\bigO\paren{\sqrt{x_i2^i \cdot \frac{1}{2^i}\sum_{j=i}^{y} x_j/2^{j-1-i}}}$. The expression under the $\bigO$ notation can be bounded as follows.

$$\sqrt{x_i2^i \cdot \frac{1}{2^i} \cdot \sum_{j=i}^{y } x_j/2^{j-1-i}}
\leq\sqrt{\paren{\sum_{j=i}^{y } x_j/2^{j-1-i}} \cdot \paren{\sum_{j=i}^{y } x_j/2^{j-1-i}}}
= \sum_{j=i}^{y } x_j/2^{j-1-i}$$ 
Therefore, the total number of edges in all reduced graphs is $\bigO(\sum_{i=1}^{y}\sum_{j=i}^{y}  x_j/2^{j-1-i})$. To give the desired bound on this sum, we look on the contribution to the sum from the point of view of the vertices from $V_j$. More precisely, the set $V_j$ contributes something only to the sums that start with such indices $i$ that $i \leq j$. The amount $V_j$ contributes to the sum starting with the specific $i$ is at most $x_j/2^{j-1-i}$. We observe that all contributions of a single set $V_j$ form a geometric series and this observation allows us to get the desired bound. The following rearrangement of the summation corresponds to this change of the point of view argument:

$$
\mathlarger{\mathlarger{\bigO}}\paren{\sum_{i=1}^y\sum_{j=i}^{y}  x_j/2^{j-1-i}} = 
\mathlarger{\mathlarger{\bigO}}\paren{\sum_{j=1}^y\sum_{i=1}^{j}  x_j/2^{j-1-i}} =
\mathlarger{\mathlarger{\bigO}}\paren{\sum_{j=1}^y 4 x_j} = \bigO\paren{n}
$$

\end{proof}

\paragraph{The final step}
The spanning forest algorithm [\cref{alg:sf_algo}] executes \cref{alg:reduce} to find a partition into several component graphs, on which we apply sparsification algorithm from \cref{lem:mst_sparsification}. As a result we obtain a set of $\bigO(n)$ inter component edges, that contains all edges of some spanning forest of the component graph. 

To obtain a spanning forest of the input graph it is sufficient to gather in the memory of a single processor
\begin{itemize}
\item all remaining inter component edges, i.e., all edges that we get as a result of an execution of the algorithm from \cref{lem:mst_sparsification} on graphs $G_i$, for all $i$, 
\item all edges used in \cref{alg:reduce} . 
\end{itemize}
Then, this processor can compute the spanning forest of the input graph $G$, by computing the spanning forest of the graph $G'$ consisting of the gathered edges. 

To see that a spanning forest of $G'$ is also a spanning forest of $G$, let us take a closer look on the edges of $G'$. The edges used in \cref{alg:reduce} provide that any two vertices of $G$ that after the execution of \cref{alg:reduce} are in a single component of $\mathcal{C}$ are connected by a path in $G'$. Adding the inter component edges to the edges used in \cref{alg:reduce} provides that any two vertices that are in a single connected component in $G$, but in different components in $\mathcal{C}$, are connected by a path in $G'$. Therefore, any two vertices that are connected by a path in $G$ are also connected by a path in $G'$. Thus, a spanning forest of $G'$ is also a spanning forest of $G$. 

\subsubsection{\cclique implementation}
So far, \cref{sec:sf_algo} discuss the building blocks of \cref{alg:sf_algo}. Here, we show that this algorithm can be executed in the \cclique model and in the \MPC model, i.e., we discuss an implementation of each step of \cref{alg:sf_algo}.

Firstly, \cref{alg:reduce} can be clearly implemented in \cclique, as it requires only communication over the edges of the input graph, and one coordinator vertex that computes the partition into connected components. After that, each vertex knows the edges which belong to $G_i$, for each $i \leq y$.

Then, we need to run several instances of sparsification algorithm from \cref{lem:mst_sparsification} in parallel. There are two parts that we need to address:
\begin{itemize}
\item an execution of the initial reduction that reduces the maximal degree to average degree, in parallel;
\item an execution of many instances of algorithm \cref{alg:korhonen} in parallel.
\end{itemize}
Those two parts are the only parts of our Spanning Forest algorithm that have a non trivial implementation. Therefore we address it only after we explain the implementation of the final step.

The final part of the \cref{alg:sf_algo} can be implemented trivially. The total number of edges of $G'$ is $2n + \bigO(n) = \bigO(n)$. Therefore, we can gather them in the memory of a single processor in a constant number of rounds. 

\paragraph{Initial reduction for a component graph, in parallel}
The initial reduction from \cref{subsec:initial} is defined for a graph, and here we need to apply it on the graphs $G_1, G_2, \dots, G_y$. To recall, the vertices of $G_1, G_2, \dots, G_y$ correspond to the connected components $\mathcal{C} = \set{C_1, C_2, \dots,}$ of $G$. Let us consider a single $G_i$. The problem we face here it that a single vertex of $G_i$ may consist of many vertices of $G$. Therefore, it is possible that a single processor does not see all the edges that are incident to a single vertex of $G_i$. 

Still, we claim that the partition of high degree vertices of $G_i$ into vertices of degrees at most $\Delta_A + 2$ can be executed almost as for a normal graph. In the first step, each vertex $v$ of $G$ counts the incident edges that belong to $G_i$, and connect $v$ to some other component of $\mathcal{C}$. Then, $v$ sends this number (let us call it the $G_i$-degree of $v$), together with a number $j$, such that $v \in C_j$, to the coordinator processor.

The processor that knows the $G_i$-degrees of all vertices of $G$, and for each vertex of $G$ knows the ID of its component in $\mathcal{C}$, can compute the degree of the vertices in $G_i$. Then, as in the case of normal graphs, for each vertex of $G_i$ the coordinator can compute the number of parts it has to be splitted into.

Let us consider a case, when the coordinator needs to split a vertex of $G_i$ that corresponds to the component $C_j$ into some number of new vertices. Let $\set{v_1, v_2, \dots, v_k}$ be the vertices of $G$ that are in $C_j$. Then, the coordinator needs to compute an assignment of vertices $\set{v_1, v_2, \dots, v_k}$ to the splitted vertices. This can be done in a greedy way.
 
To picture the greedy assignment, we can imagine that a vertex of degree $\delta$ is a block of height $1$ and length $\delta$. Then, we put together the blocks of all vertices $\set{v_1, v_2, \dots, v_k}$, creating one long block of height $1$ and length that is the degree of $C_j$ in $G_i$. Then, we split this long block into pieces of length $\Delta_A$ (the last one may be shorter). Each of the pieces corresponds to a single vertex of the low degree graph that we want to compute. 

Now, the coordinator needs only to notify all vertices of $G$ which pieces overlap with their block. Furthermore, for the first and last overlapping piece, the coordinator needs to specify the size of the overlap. 

For each vertex, the information (ID and overlap size) about the first piece, the last piece and the number of pieces that are neither first or last, can be encoded on $4$ messages (ID of the first piece, overlap with the first piece, ID of the last piece, overlap with the last piece). Therefore, the coordinator has to send at mist $\bigO(1)$ messages per vertex, and $\bigO(n)$ messages in total. 

Then, the vertices of $G$ locally assign particular edges to particular pieces. For each edge, the processor exchanges the assignment with the processor that holds the other endpoint of the edge. This way, for each edge $e$ both processors holding $e$ know the new IDs of the endpoints of $e$. Therefore, it is possible to determine $E_{i,j} \ni e$ which is enough to run \cref{alg:korhonen}. 

To show that this algorithm can be executed in parallel, for all $G_i$ simultaneously, it is enough to show that any vertex does not need to send too many messages. Since there are only $\bigO(\log n)$ instances, and a single vertex sends at most $\bigO(1)$ messages to a coordinator, the communication per vertex is $\bigO(\log n)$. Furthermore, we can bound the total communication by $\bigO(m)$ -- this follows from that a message is sent by a vertex $v$ to the coordinator of the $i$th instance only if there is an edge incident to $v$ in $G_i$. Since $G_i$ are edge disjoint, and we have at most two messages per edge, the total number of messages is $\bigO(m)$.

\paragraph{\cref{alg:korhonen} in parallel}
To explain that \cref{alg:korhonen} can be executed for all $G_i$ simultaneously, in parallel, we use a similar argument as for the initial reduction. The number of messages that are send by a single vertex of $G$ in the instance of \cref{alg:korhonen} for the graph $G_i$ is proportional to its $G_i$-degree. Therefore, even though a single vertex may participate in many instances of the algorithm, the total number of messages it sends cannot be larger than the sum of degrees of all vertices that it simulates. Since the instances are edge disjoint, it is $\bigO(n)$ per vertex. 

On the coordinator side, nothing changes with respect to the original analysis of \cref{alg:korhonen}, as each coordinator receives $\bigO(n)$ edges. Furthermore, since the initial reduction increases the number of edges at most by some constant factor, the number of coordinators remains $\bigO(m/n)$. 

To summarize, we have that any processor sends and receives $\bigO(n)$ messages and the total number of messages is $\bigO(m)$. Therefore, the parallel execution of \cref{alg:korhonen} can be carried out in \cclique and its communication complexity is $\bigO(m)$.

\section{The Algorithm for the Minimum Spanning Tree problem}\label{sec:msf_algo}
In this section we show that our algorithm for the Spanning Tree problem can be applied to the Minimum Spanning Forest, proving \cref{thm:msf_algo}.

\begin{thm:msf_algo}
The Minimum Spanning Tree problem can be solved deterministically in $\bigO(1)$ rounds of the \cclique model, using $\bigO(m)$ communication, $m$ is the number of edges of the input graph.
\end{thm:msf_algo}

We prove \cref{thm:msf_algo} using \cref{lem:mst_sparsification} and \cref{thm:sf_algo}. Firstly, let us recall the reduction from a single instance of the \MST problem to several instances of the Connected Components problem \cite{Hegeman:2015:TOB:2767386.2767434}. More precisely, we give a variant of this reduction that is deterministic, as it is based on the deterministic sparsification algorithm from \cref{lem:mst_sparsification} rather than on the randomized sparsification technique by Karger et al. \cite{Karger:1995:RLA:201019.201022}. 

\begin{lemma}\label{lem:msf_to_cc}
  There is a deterministic, $\bigO(1)$ round \cclique algorithm that reduces the problem of computing the \MST of an $n$ vertex, $m$ edge graph $G$ to $\bigO(\sqrt{m/n})$ independent instances of the Connected Components problem, such that the total number of edges in obtained instances is $\bigO(m)$.
\end{lemma}
\begin{proof}
On the top level, the reduction relies on some properties of Kruskal's algorithm for the MST problem. In particular, Kruskal's algorithm uses the following property. An edge $e$ is added to MST iff its endpoints belong to different connected components of the graph containing only the edges of $G$ that are lighter than $e$. 

The idea proposed by Hegeman et al. \cite{Hegeman:2015:TOB:2767386.2767434} is as follows. Firstly, we sort the edges by weight, using $\bigO(1)$ sorting algorithm by Lenzen \cite{Lenzen:2013:ODR:2484239.2501983}\footnote{This step uses $\Theta(n^2)$ communicates, but it is only because it uses a routing protocol; if we have routing for free, the communication complexity is proportional to the number of sorted elements}. Then, we split the sorted sequence of edges into sets $E_1, E_2, \dots, E_{m/n}$, each of size $n$. For each $i \in [1,m/n]$ we compute the connected components $\mathcal{C}_i = (C_1, C_2, \dots)$ of a graph with edges $\bigcup_{j=1}^{i-1} E_j$. To do so, we use a Connected Components algorithm, in parallel. Then a single processor can gather $\mathcal{C}_i$ and $E_i$ in the local memory, and simulate the steps of Kruskal's algorithm on $E_i$. That is, given the connected components of graph consisting of the edges $\bigcup_{j=1}^{i-1} E_j$, this processor can process all edges of $E_i$, in the order from the lightest to the heaviest. For each edge $e$ that is processed, the processor knows the connected components of the graph containing only the edges of $G$ that are lighter than $e$. Therefore, this processor can decide whether $e$ belongs to the MST of $G$.

The only issue with using this approach in a straightforward way is that starting from $m$ edge graph, this gives $\frac{m}{n}$ instances of the Connected Components problem with total size that could be $\Theta((m/n)^2) \cdot \Theta(n)$. To bypass this issue, Hegeman et al. used a random sampling approach proposed by Karger, Klein, and Tarjan \cite{Karger:1995:RLA:201019.201022} that can be used to reduce a single instance of the \MST problem to two instances of this problem that have to be executed one after the other, each of size $\bigO(\sqrt{mn})$. Here, we replace the randomized sparsification algorithm with the deterministic algorithm from \cref{lem:mst_sparsification}.

For an input graph with $\bigO(\sqrt{mn})$ edges, the reduction by Hegeman et al. \cite{Hegeman:2015:TOB:2767386.2767434} gives $\bigO(\sqrt{m/n})$ instances of the Connected Components problem with $\bigO((\sqrt{m/n})^2) \cdot \Theta(n) = \bigO(m)$ edges in total. Still, in the memory of the processors, we have only sets $E_i$, for $i \in [1, \sqrt{m/n}]$, and some of them participate in many instances of the Connected Components problem. In order to make it clear that we can easily solve those several instances in parallel, we show that we can duplicate some of the sets $E_i$, so that each instance of the Connected Components problem has its own copy of $E_i$. Our goal is to have $k$ copies of an edge that appears in $k$ instances. This allows us to provide a rather clean way of running the Spanning Forest algorithm in parallel.

To perform the duplication efficiently, we assign $k-1$ helper processors to each processor $P$ that holds a set of $\Theta(n)$ edges that should participate in $k$ instances. Let $P_k$ denotes the set consisting of this processor and its helper processors. Duplication can be executed in two stages. In the first stage, each processor of $P_k$ receives from $P$ a part of the set of edges of size $\Theta(n/k)$. In the second stage, each processor from $P_k$ sends the received part to all other processors in $P_k$. 

The first step of duplication can be executed as $P$ sends $\bigO(n)$ edges in total, and each processor in $P_k$ receives only $\bigO(n/k)$ messages. Then, in the second step, each processor from $P_k$ sends $\card{P_k} \cdot \bigO(n/k) = \bigO(n)$ messages. Finally, each processor in $P_k$ receives the set of all edges that was stored in the memory of $P$, and its size is $\Theta(n)$. Furthermore, since after the duplication each helper processor keeps $\Theta(n)$ edges and the total number of edges is 
$\bigO(m)$, we need only $\bigO(m/n)$ helper processors in total.
\end{proof}

\subsection{A parallel execution of many instances of the algorithm for Spanning Forest problem}
In this subsection, we show that we can solve the instances of Connected Components problem obtained by the reduction from \cref{lem:msf_to_cc} using the Spanning Forest algorithm from \cref{thm:sf_algo} in parallel, which completes the proof of \cref{thm:msf_algo}.

\paragraph{A short argument} The Spanning Forest algorithm from \cref{thm:sf_algo} is an \MPC algorithm (see \cite{DBLP:journals/corr/abs-1802-10297} and our analysis of the communication complexity throughout the paper). The reduction that reduces a single instance of the MST problem to several instances of the Connected Components, $\mathit{CC}_1, \mathit{CC}_2, \dots, \mathit{CC}_{\sqrt{\smash[b]{m/n}}}$. The \cclique model can execute the \MPC Spanning Forest algorithm on all those instances in $\bigO(1)$ rounds. 

The more precise explanation is as follows. Let $m_i$ be the size of the instance $\mathit{CC}_i$. An \MPC spanning forest algorithm that solves $\mathit{CC}_i$ needs only $\bigO(m_i)$ global communication and it can be simulated in \cclique on $\bigO(m_i/n)$ processors, with the help of Lenzen's routing protocol \cite{Lenzen:2013:ODR:2484239.2501983}. Furthermore, we can simulate several instances of the Spanning Forest algorithm, just by assigning disjoint sets of processors to different instances of the Spanning Forest algorithm. This is because a single instance of the Lenzen's routing protocol can handle \MPC-like communication for all instances of the Spanning Forest algorithm, simultaneously.

Therefore, as long as the total number of processors remains $\bigO(n)$ and the total global communication is $\bigO(n^2)$, all instances of the \MPC Spanning Forest algorithm can be executed simultaneously. Furthermore, the total global communication is proportional to the global communication of all instances. Therefore, we can solve the instances $\mathit{CC}_1, \mathit{CC}_2, \dots, \mathit{CC}_{\sqrt{\smash[b]{m/n}}}$ using  $\sum_{i=1}^{\sqrt{\smash[b]{m/n}}} \bigO(m_i) = \bigO(m)$ global communication.

\paragraph{Simulation in \cclique}
In the remaining part of this section, for the sake of completeness, we give some details of the parallel execution of the Spanning Forest algorithm from \cref{thm:sf_algo} in the \cclique model, without referring to the simulation of \cclique algorithms in the \MPC model.

We are given sets of edges of $\bigO(\sqrt{m/n})$ graphs, with $\bigO(m)$ edges in total. Our goal is to compute a representation of each of those graphs that is a \emph{vertex partition}. More precisely, we want that each processor instead of getting an arbitrary set of edges, gets a set of vertices and all edges incident to them (that is a \emph{vertex partition} of the input). For each edge the processor needs to know the ID of a processor holding the other endpoint. Furthermore, we want to partition the vertices in such a way that executing the Spanning Forest algorithm for all graphs in parallel can be efficiently simulated by the processors of the \cclique. Basically, the goal is to partition the vertices of all instances in such a way that:
\begin{itemize}
  \item each processor simulates vertices that have $\bigO(n)$ incident edges in total,
  \item a processor simulating a particular vertex $v$ in any instance of the Connected Components problem, knows all the edges incident to $v$ in the considered instance,
  \item a processor simulating a particular vertex $v$ in any instance of the Connected Components for each edge $\set{v,u}$ knows the ID of the processor simulating $u$ in that instance of the Connected Components problem.
\end{itemize}

Clearly, those three properties guarantee that we can execute communication between the neighbours in the simulated graphs. Furthermore, it also implies that communication with coordinators in \cref{alg:korhonen} and \cref{alg:reduce} can be executed efficiently. 

In \cref{alg:korhonen} the number of messages that a single vertex $v$ sends is proportional to its degree. Since, the sum of degrees of all vertices simulated by a single processor is $\bigO(n)$, the total number of the messages that the processor needs to send is also $\bigO(n)$. In \cref{alg:reduce} a vertex communicates with the coordinator, only if it has a non zero degree in the instance that uses this coordinator. Since the instances are edge disjoint, the total number of messages send by one vertex is always no larger than the sum of degrees of simulated vertices, which is $\bigO(n)$. 

\paragraph{Partition of simulated vertices}
To obtain a partition of simulated vertices allowing the parallel execution of the Spanning Forest algorithm, we do the following. We start by copying each edge $\set{u,v}$ twice, we create one copy for an edge \emph{outgoing} from $u$, and one copy for an edge \emph{outgoing} from $v$. Then, we sort this set of edges, to assure that all edges \emph{outgoing} from a single vertex $v$ are in the memory of a single processor (and we have that for all $v$ simultaneously). 

We can sort the edges using the $\bigO(1)$ round sorting algorithm \cite{Lenzen:2013:ODR:2484239.2501983}, but as a result of sorting, we do not have guarantee that for each $v$ we see all the edges in the memory of a single processor. If a processor with ID $x$ does not see all edges that are incident to some vertex $v$, those edges are in the memory of a processor with ID $x \pm 1$. Therefore, a a processor with ID $x$ can communicate with processors with ID $x \pm 1$; if some two processors have the edges that are \emph{outgoing} from a single vertex $v$ the one with smaller ID can send them to the one with larger ID, which can be done in $\bigO(1)$ rounds.

From now on, a processor that holds the edges \emph{outgoing} from $v$ simulates $v$, and the remaining work we have to do is to find which processors simulate the other endpoints of the edges incident to $v$. To do so, we again use a sorting algorithm. To each of two copies of an edge we attach additional information stating which processor simulates one of the endpoints. Then, we sort the edges by the endpoints. As a result, some processor $p$ sees two copies of the edge, each having $\ID$ of a processor simulating one endpoint. Then, $p$ notifies the processors simulating the endpoints of an edge, what is the $\ID$ of a processor simulating the other endpoint.

Since as a result of sorting each processor gets $\bigO(n)$ edges, it has to send at most $\bigO(n)$ notifications. Furthermore, each processor simulates vertices of degrees that sum to $\bigO(n)$, therefore, it need to receive $\bigO(n)$ notifications. Thus, this step can be executed in $\bigO(1)$ rounds.

\subsection*{Acknowledgment}
We are grateful to Mohsen Ghaffari and Tomasz Jurdziński for all discussions on the MST problem and the $k$--out contraction technique, as those pushed us in the right direction. We also thank for all their comments that helped to improve the clarity of this paper.

\bibliographystyle{alpha} 
\bibliography{ref}

\end{document}